\documentclass[11pt,letterpaper]{article} 
\usepackage[T1]{fontenc}
\usepackage[english]{babel}
\usepackage[margin=.95in]{geometry}
\usepackage{amsmath,amsthm,amssymb}
\usepackage{thmtools,thm-restate}
\usepackage{mathtools}
\usepackage{url}
\usepackage[textsize=scriptsize]{todonotes}
\usepackage[noend]{algpseudocode}
\usepackage{algorithm}
\usepackage[hidelinks]{hyperref}

\newcommand{\Viterbi}{\textsc{Viterbi Path}}
\newcommand{\Walk}{\textsc{Shortest Walk}}
\newcommand{\APSP}{\textsc{All-Pairs Shortest Paths}}
\newcommand{\kClique}{\textsc{Min-Weight $k$-Clique}}

\newcommand{\MinTriangle}{\textsc{Minimum Triangle}}
\declaretheorem[]{theorem}
\declaretheorem[sibling=theorem]{definition}

\declaretheorem[sibling=theorem]{claim}
\declaretheorem[]{remark}

\newcommand{\R}{\mathbb{R}}

\newcommand{\eps}{\varepsilon}

\DeclarePairedDelimiter{\ceil}{\lceil}{\rceil}

\DeclareMathOperator*{\argmin}{arg\,min}
\DeclareMathOperator*{\argmax}{arg\,max}

\begin{document}

		\begin{titlepage}
\title{Improving Viterbi is Hard:\\Better Runtimes Imply Faster Clique Algorithms}
\date{}
\author{
	Arturs Backurs\\ MIT 
	\and Christos Tzamos\\ MIT
	}
\clearpage
\maketitle
\thispagestyle{empty}
\begin{abstract}
The classic algorithm of Viterbi computes the most likely path 
in a Hidden Markov Model (HMM) that results in a sequence of observations.
It runs in time $O(Tn^2)$ given a sequence of $T$ observations from a HMM with $n$ states.
Despite significant interest in the problem and prolonged effort by different communities,
no known algorithm achieves more than a polylogarithmic speedup.

In this paper, we explain this difficulty by providing matching conditional lower bounds. We show that the Viterbi runtime is 
tight up to subpolynomial factors even when the number of distinct observations is small. 
Our lower bounds are based on assumptions that the best known algorithms for the 
All-Pairs Shortest Paths problem (APSP) and for the Max-Weight k-Clique problem in edge-weighted graphs are essentially optimal.

Finally, using a recent algorithm by Green Larsen and Williams for online Boolean matrix-vector multiplication, we get a $2^{\Omega(\sqrt {\log n})}$ speedup for
the Viterbi algorithm when there are few distinct transition probabilities in the HMM.
\end{abstract}
\end{titlepage} 	
	\section{Introduction}

A Hidden Markov Model (HMM) is a simple model that describes a random process for generating a sequence of observations. A random walk is performed on an underlying graph (Markov Chain) and, at each step, an observation is drawn from a probability distribution that depends only on the current state (the node in the graph). 
HMMs are a fundamental statistical tool and have found wide applicability in a number of fields such as Computational Biology~\cite{haussler1996generalized,krogh2001predicting,petersen2011signalp}, Signal Processing~\cite{gales1998maximum,huang1990hidden,kupiec1992robust}, Machine Learning and Computer Vision~\cite{starner1998real, agazzi1993hidden, zhang2001segmentation}.

One of the most important questions in these applications is computing the most likely sequence of states visited by the random walk in the HMM given the sequence of observations. 
This problem, known as the \emph{Viterbi Path} problem, is named after Andrew Viterbi who proposed an algorithm~\cite{viterbi1967error}, known as the \emph{Viterbi algorithm}, that computes the solution in $O(Tn^2)$ time for any HMM with $n$ states and an observation sequence of length $T$.

The quadratic dependence of the algorithm's runtime on the number of states is a long-standing bottleneck that limits its applicability to problems with large state spaces, particularly when the number of observations is large. A lot of effort has been put into improving the Viterbi algorithm to lower either the time or space complexity. 

Many works achieve speedups by requiring structure in the input, either explicitly by considering restricted classes of HMMs~\cite{felzenszwalb2004fast} or implicitly by using  heuristics that improve runtime in certain cases~\cite{esposito2009carpediem, kaji2010efficient}. For the general case, in \cite{lifshits2009speeding, mahmud2011speeding} it is shown how to speed up the Viterbi algorithm by $O(\log n)$ when the number of distinct observations is constant using the Four Russians method or similar ideas. More recently, in \cite{cairo2016decoding}, the same logarithmic speed-up was shown to be possible for the general case. Despite significant effort, only logarithmic improvements are known other than in very special cases. In contrast, the memory complexity can be reduced to almost linear in the number of states without significant overhead in the runtime~\cite{grice1997reduced, tarnas1998reduced, churbanov2008implementing}.

In this work, we attempt to explain this apparent barrier for faster runtimes by giving evidence of the inherent hardness of the Viterbi Path problem. 
In particular, we show that getting a polynomial speedup\footnote{Getting an algorithm running in time, say $O(Tn^{1.99})$.} would imply a breakthrough for fundamental graph problems. 
Our lower bounds are based on standard hardness assumptions for the All-Pairs Shortest Paths and the Min-Weight $k$-Clique problems and apply even in cases where the number of distinct observations is small.

We complement our lower bounds with an algorithm for Viterbi Path that achieves speedup $2^{\Omega(\sqrt{\log n})}$ when there are few distinct transition probabilities in the underlying HMM.

\paragraph{Our results and techniques} 

Our first lower bound shows that the Viterbi Path problem cannot be computed in time $O(Tn^2)^{1-\eps}$ for a constant $\eps > 0$ unless the APSP conjecture is false. The APSP conjecture states that there is no algorithm for the All-Pairs Shortest Paths problem that runs in truly subcubic\footnote{Truly subcubic means $O(n^{3-\delta})$ for constant $\delta > 0$.} time in the number of vertices of the graph. We obtain the following theorem:

\begin{restatable}{theorem}{viterbiapsp}
\label{thm:viterbiapsp}
	The $\Viterbi$ problem requires $\Omega(Tn^2)^{1-o(1)}$ time assuming the APSP Conjecture.
\end{restatable}

The proof of the theorem gives a reduction from the All-Pairs Shortest Paths problem to the Viterbi Path problem. This is done by encoding the weights of the graph of the APSP instance as transition probabilities of the HMM or as probabilities of seeing observations from different states. The proof requires a large alphabet size, i.e. a large number of distinct observations, which can be as large as the number of total steps $T$. 

A natural question question to ask is if there is a faster algorithm that solves the Viterbi Path problem when the alphabet size is much smaller than $T$, say when $T=n^2$ and the alphabet size is $n$. We observe that in such a case, the input size to the Viterbi Path problem is only $O(n^2)$. We only need to specify the transition probabilities of the HMM, the probabilities of each observation in each state and the sequence of observations which take $O(n^2)$ space in total. The Viterbi algorithm in this setting runs in $\Theta(T n^2) = \Theta(n^4)$ time. Showing a matching APSP based lower bound seems difficult because the runtime in this setting is quadratic in the input size while the APSP conjecture gives only $N^{1.5}$ hardness for input size $N$. To our best knowledge, all existing reduction techniques cannot achieve such an amplification of hardness. In order to get a lower bound for smaller alphabet sizes, we need to use a different hardness assumption. 

For this purpose, we consider the $k$-Clique conjecture. It is a popular hardness assumption which states that it is not possible to compute a minimum weight $k$-clique on an edge-weighted graph with $n$ vertices in time $O(n^{k-\eps})$ for constant $k$ and $\eps > 0$. With this assumption, we are able to extend Theorem~\ref{thm:viterbiapsp} and get the following lower bound for the Viterbi Path problem on very small alphabets:

\begin{restatable}{theorem}{smallalpha}
	\label{thm:smallalpha}
	For any $C,\eps>0$, the $\Viterbi$ problem on $T = \Theta(n^{C})$ observations from an alphabet of size $O(n^{\eps})$ requires $\Omega(Tn^2)^{1-o(1)}$ time assuming the $k$-Clique Conjecture for $k = \ceil{\frac C \eps}+2$.
\end{restatable}

To show the theorem, we perform a reduction from the Min-Weight $k$-Clique problem. Given a Min-Weight $k$-Clique instance, we create an HMM with two special nodes, a start node and an end node, and enforce the following behavior of the optimal Viterbi path: Most of the time it stays in the start or end node, except for a small number of steps, during which it traverses the rest of the graph to move from the start to the end node. The time at which the traversal happens corresponds to a clique in the original graph of the Min-Weight $k$-Clique instance. We penalize the traversal according to the weight of the corresponding $k$-clique and thus the optimal path will find the minimum weight $k$-clique. 
Transition probabilities of the HMM and probabilities of seeing observations from different states encode edge-weights of the Min-Weight $k$-Clique instance. Further, we encode the weights of smaller cliques into the sequence of observations according to the binary expansion of the weights.

Our results of Theorems~\ref{thm:viterbiapsp} and~\ref{thm:smallalpha} imply that the Viterbi algorithm is essentially optimal even for small alphabets. We also study the extreme case of the Viterbi Path problem with unary alphabet where the only information available is the total number of steps $T$. We show a surprising behavior: when $T \le n$ the Viterbi algorithm is essentially optimal, while there is a simple much faster algorithm when $T > n$. See Section~\ref{sec:noobs} for more details.

We complement our lower bounds with an algorithm for Viterbi Path that achieves speedup $2^{\Omega(\sqrt{\log n})}$ when there are few distinct transition probabilities in the underlying HMM. Such a restriction is mild in applications where one can round the transition probabilities to a small number of distinct values.

\begin{restatable}{theorem}{algorithm}
	\label{thm:algorithm}
	When there are fewer than $2^{\eps \sqrt {\log n} }$  distinct transition probabilities for a constant $\eps > 0$, there is a ${Tn^2} / {2^{\Omega(\sqrt {\log n}) }}$ randomized algorithm for the $\Viterbi$ problem that succeeds whp.
\end{restatable}

We achieve this result by developing an algorithm for online $(\min,+)$ matrix-vector multiplication for the case when the matrix has few distinct values. Our algorithm is based on a recent result for online Boolean matrix-vector multiplication by Green Larsen and Williams~\cite{larsen2016faster}.

Finally, we provide an algorithm that runs in $O(T n^{1.99})$ \emph{non-deterministic} time. This provides an evidence that lower bounds for Viterbi Path based on Strong Exponential Time Hypothesis (SETH) are not possible. See Section~\ref{sec:nondeterministic} for more details.

\paragraph{Hardness assumptions}
There is a long list of works showing conditional hardness for various problems based on the All-Pairs Shortest Paths problem hardness assumption~\cite{roditty2004dynamic,williams2010subcubic,abboud2014popular,abboud2015subcubic,abboud2015matching}.
Among other results, \cite{williams2010subcubic} showed that finding a triangle of minimum weight in a weighted graph is equivalent to the All-Pairs Shortest Paths problem meaning that a strongly subcubic algorithm for the $\MinTriangle$ problem
implies a strongly subcubic algorithm for the All-Pairs Shortest Paths problem and the other way around. 
Computing a min-weight triangle is a special case of the problem of computing a min-weight $k$-clique in a graph for a fixed integer $k$. This is a very well studied computational problem and despite serious efforts, the best known algorithm for this problem still runs in time $O(n^{k-o(1)})$, which matches the runtime of the trivial algorithm up to subpolynomial factors. The assumption that there is no $O(n^{k-\eps})$ time algorithm for this problem, has served as a basis for showing conditional hardness results for several problems on sequences \cite{abboud2015if, abboud2014consequences} and computational geometry~\cite{backurs2016tight}. 	
	\section{Preliminaries}

\paragraph*{Notation} For an integer $m$, we denote the set $\{1,2,\ldots,m\}$ by $[m]$.

\begin{definition}[Hidden Markov Model] \label{def:HMM}
	A \emph{Hidden Markov Model} (HMM) consists of a directed graph with $n$ distinct hidden states $[n]$ with transition probabilities $\tilde A(u,v)$ of going from state $u$ to state $v$. In any given state, there is a probability distribution of symbols that can be observed and $\tilde B(u,s)$ gives the probability of seeing symbol $s$ on state $u$. The symbols come from an alphabet $[\sigma]$ of size $\sigma$. An HMM can thus be represented by a tuple $(\tilde A, \tilde B)$.
\end{definition}

\subsection{The Viterbi Path Problem} 
Given an HMM and a sequence of $T$ observations, the Viterbi algorithm \cite{viterbi1967error} outputs a sequence of $T$ states that is most likely given the $T$ observations. More precisely, let $S=(s_1, \ldots, s_T)$ be the given sequence of $T$ observations where symbol $s_t \in [\sigma]$ is observed at time $t=1, \ldots, T$. Let $u_t \in [n]$ be the state of the HMM at time $t=1, \ldots, T$. The Viterbi algorithm finds a state sequence $U=(u_0,u_1, \ldots, u_T)$ starting at $u_0=1$ that maximizes $\Pr[U|S]$. The problem of finding the sequence $U$ is known as the \emph{Viterbi Path} problem. In particular, the Viterbi Path problem solves the optimization problem
$$\argmax_{u_0=1,u_1,\dots,u_T} \prod_{t=1}^T \left[ \tilde A(u_{t-1},u_{t}) \cdot \tilde B(u_{t},s_t) \right].$$

\paragraph{\iffalse$\boldsymbol{\log}$\fi Log probabilities} Instead of working directly with probabilities and products of probabilities, it is convenient apply the $\log$ function in the above problem and work with sums instead. In practical applications, $\log$ probabilities are often used for numerical stability since that helps avoid underflows~\cite{young1997htk,amengual1998efficient,li2009design,lee2007design,huang2001spoken}. For this reason in this paper we set $A(u,v)=-\log \tilde A(u,v)$ and $B(u,s)=-\log \tilde B(u,s)$ and focus our attention on the additive version of the Viterbi Path problem which we define below.

\begin{definition}[Viterbi Path Problem] \label{def:viterbipath}
The $\Viterbi$ problem is specified by a tuple $(A,B,S)$ where $A$ and $B$ are $n \times n$ and $n \times \sigma$ matrices, respectively, of non-negative real valued entries and $S = (s_1, \ldots, s_T)$ is a sequence of $T = n^{\Theta(1)}$ observations $s_1, \ldots, s_T \in [\sigma]$ over an alphabet of size $\sigma$.
Given an instance $(A,B,S)$ of the $\Viterbi$ problem, our goal is to output a sequence of  vertices $u_0, u_1, \ldots, u_T \in [n]$ with $u_0 = 1$ that solves
$$\argmin_{u_0=1,u_1,\dots,u_T} \sum_{t=1}^T \left[  A(u_{t-1},u_{t}) +  B(u_{t},s_t) \right].$$
\end{definition}

A simpler special case of the $\Viterbi$ problem asks to compute the most likely path of length $T$ without any observations. 
\begin{definition}[Shortest Walk Problem]
Given an integer $T$ and a weighted directed graph (with possible self-loops) on $n$ vertices with edge weights specified by a matrix $A$, the $\Walk$ problem asks to compute a sequence of vertices $u_0=1,u_1,\dots,u_T \in [n]$ that solves
$$\argmin_{u_0=1,u_1,\dots,u_T} \sum_{t=1}^T A(u_{t-1},u_{t}).$$
\end{definition}

It is easy to see that the $\Walk$ problem corresponds to the $\Viterbi$ problem when $\sigma = 1$ and $B(u,1)=0$ for all $u \in [n]$.

\begin{remark}
	Working with probabilities is equivalent to working with the $\log$ probabilities if we are considering a real RAM model where we are allowed to perform $\log$ and exponentiation operations. In the word RAM model these variants are different. As it was discussed above, the computation with the $\log$ probabilities is numerically more stable.
	
	Moreover, we may assume that $\log$ probabilities in matrices $A$ and $B$ are arbitrary positive numbers without the restriction that the corresponding probabilities must sum to $1$. See Appendix~\ref{app:sumprob} for discussion. 
\end{remark}

\subsection{Hardness assumptions}

We use the hardness assumptions of the following problems.

\begin{definition}[$\APSP$ (APSP) problem]
	\label{def_apsp}
Given an undirected graph $G=(V,E)$ with $n$ vertices and polynomially bounded positive integer weights on the edges, find the shortest path between $u$ and $v$ for every $u,v \in V$.
\end{definition}

\noindent The \emph{APSP conjecture} states the $\APSP$ problem requires $\Omega(n^3)^{1-o(1)}$ time.

\begin{definition}[$\kClique$ problem]
	\label{def_clique}
	Given an integer $k$ and a complete graph $G=(V,E)$ with $n$ vertices and polynomially bounded positive integer weights on the edges, output the minimum total edge-weight of a $k$-clique in the graph.
\end{definition}
For any fixed constant $k$, the best known algorithm for the $\kClique$ problem runs in time $O(n^{k-o(1)})$ and the \emph{$k$-Clique conjecture} states that it requires $\Omega(n^k)^{1-o(1)}$ time.

For $k=3$, the $\textsc{Min-Weight $3$-Clique}$ problem (also known as the $\MinTriangle$ problem) asks to find the minimum weight triangle in a graph and it is known that the $3$-Clique conjecture is equivalent to the APSP conjecture~\cite{williams2010subcubic}.

We often use the following variant of the $\kClique$ problem:

\begin{definition}[$\kClique$ problem for $k$-partite graphs] \label{kpartite}
Given an integer $k$ and a complete $k$-partite graph $G=(V_1 \cup \ldots \cup V_k, \ E)$ with $|V_i|=n_i$ and polynomially bounded positive integer weights on the edges, output the minimum total edge-weight of a $k$-clique in the graph.

\end{definition}

If for all $i,j$ we have that $n_i = n_j^{O(1)}$, it can be shown that the $\kClique$ problem for $k$-partite graphs requires $\Omega\left(\prod_{i=1}^k n_i \right)^{1-o(1)}$ time assuming the $k$-Clique conjecture. To see this, assume, for example, that there is an $O(n^{1+1/2+1/4-\eps})$ algorithm $\mathcal{A}$ that finds a minimum weight $3$-clique in a $3$-partite graph with $n_1=n^{1/4}$, $n_2=n^{1/2}$ and $n_3=n$. We can use $\mathcal{A}$ to find a $3$-clique in a graph $G=(V, \ E)$ where $|V|=n$, as follows: Let $\mathcal{V}^1$ be a partition of $V$ into $n^{3/4}$ sets of size $n^{1/4}$ and let $\mathcal{V}^2$ be a partition of $V$ into $n^{1/2}$ sets of size $n^{1/2}$. For all $V_1 \in \mathcal{V}^1$, $V_2 \in \mathcal{V}^2$ and $V_3 = V$, we create a $3$-partite graph $G' = (V_1 \cup V_2 \cup V_3, E')$ by adding edges corresponding to the edges of graph $G$ between nodes across partitions and find the minimum clique in the graph $G'$ using $\mathcal{A}$. Computing the minimum clique out of all the graphs we consider gives the solution to the $3$-clique instance on $G$. The total runtime is $n^{3/4 + 1/2} \cdot O(n^{1+1/2+1/4-\eps}) = O(n^{3-\eps})$ which would violate the $k$-Clique conjecture.

	\section{Hardness of $\Viterbi$}

We begin by presenting our main hardness result for the $\Viterbi$ problem.

\viterbiapsp*

We will perform a reduction from the $\MinTriangle$ problem to the $\Viterbi$ problem. In the instance of the $\MinTriangle$ problem, we are given a $3$-partite graph $G=(V_1 \cup V_2 \cup U, \ E)$ such that $|V_1| = |V_2| = n$, $|U| = m$.  We want to find a triangle of minimum weight in the graph $G$. To perform the reduction, we define a weighted directed graph $G'=(\{1,2\} \cup V_1 \cup V_2, \ E')$. $E'$ contains all the edges of $G$ between $V_1$ and $V_2$, directed from $V_1$ towards $V_2$, edges from $1$ towards all nodes of $V_1$ of weight $0$ and edges from all nodes of $V_2$ towards $2$ of weight $0$. We also add a self-loops at nodes $1$ and $2$ of weight $0$. 

We create an instance of the $\Viterbi$ problem $(A,B,S)$ as follows:
\begin{itemize}
\item Matrix $A$ is the weighted adjacency matrix of $G'$ that takes value $+\infty$ (or a sufficiently large integer) for non-existent edges and non-existent self-loops. 
\item The alphabet of the HMM is $U \cup \{ \bot, \bot_F \}$ and thus matrix $B$ has $2n+2$ rows and $\sigma = m+2$ columns. For all $v \in V_1 \cup V_2$ and $u \in U$, $B(v,u)$ is equal to the weight of the edge $(v,u)$ in graph $G$. Moreover, for all $v \in V_1 \cup V_2$, $B(v,\bot) = +\infty$ (or a sufficiently large number) and for all $v \in V_1 \cup V_2 \cup \{1\}$, $B(v,\bot_F) = +\infty$. Finally, all remaining entries corresponding to node $1$ are $0$.
\item Sequence $S$ of length $T=3m+1$ is generated by appending the observations $u$, $u$ and $\bot$ for all $u \in U$ and adding a $\bot_F$ observation at the end. 
\end{itemize}

\noindent Given the above construction, the theorem statement follows directly from the following claim.

\begin{claim}
The weight of the solution to the $\Viterbi$ instance is equal to the weight of the minimum triangle in the graph $G$.
\end{claim}

\begin{proof}
The optimal path for the $\Viterbi$ instance begins at node $1$. It must end in node $2$ since otherwise when observation $\bot_F$ arrives we collect cost $+\infty$. Similarly, whenever an observation $\bot$ arrives the path must be either on node $1$ or $2$. Thus, the path first loops in node $1$ and then goes from node $1$ to node $2$ during three consecutive observations $u$, $u$ and $\bot$ for some $u \in U$ and stays in node $2$ until the end. Let $v_1 \in V_1$ and $v_2 \in V_2$ be the two nodes visited when moving from node $1$ to node $2$. The only two steps of non-zero cost are:
\begin{enumerate}
\item Moving from node $1$ to node $v_1$ at the first observation $u$. This costs $A(1,v_1)+B(v_1,u) = B(v_1,u)$. 
\item Moving from node $v_1$ to node $v_2$ at the second observation $u$. This costs $A(v_1,v_2)+B(v_2,u)$. 
\end{enumerate}
Thus, the overall cost of the path is equal to $B(v_1,u)+A(v_1,v_2)+B(v_2,u)$, which is equal to the weight of the triangle $(v_1,v_2,u)$ in $G$. Minimizing the cost of the path in this instance is therefore the same as finding the minimum weight triangle in $G$.
\end{proof}

 	\section{Hardness of $\Viterbi$ with small alphabet}

The proof of Theorem~\ref{thm:viterbiapsp} requires a large alphabet size, which can be as large as the number of total steps $T$.
In this section, we show how to get a lower bound for the $\Viterbi$ problem on alphabets of small size by using a different hardness assumption.

\smallalpha*

\paragraph{Reduction}
Throughout the proof, we set $p = \ceil{\frac C \eps}$ and $\alpha = \frac{C}{p} \le \eps$.
	
We will perform a reduction from the $\kClique$ problem for $k=p+2$ to the $\Viterbi$ problem. In the instance of the $\kClique$ problem, we are given a $k$-partite graph $G=(V_1 \cup V_2 \cup U_1 \ldots \cup U_p, \ E)$ such that $|V_{1}| = |V_{2}| = n$ and $|U_1| = \ldots = |U_p| = m = \Theta(n^\alpha)$. We want to find a clique of minimum weight in the graph $G$. Before describing our final $\Viterbi$ instance, we first define a weighted directed graph $G'=(\{1,2,3\} \cup V_1 \cup V_2, \ E')$ similar to the graph in the proof of Theorem~\ref{thm:viterbiapsp}. $E'$ contains all the edges of $G$ between $V_{1}$ and $V_{2}$, directed from $V_1$ towards $V_2$, edges from node $1$ towards all nodes in $V_1$ of weight $0$ and edges from all nodes in $V_2$ towards node $2$ of weight $0$. We also add a self-loop at nodes $1$ and $3$ of weight $0$ as well as an edge of weight 0 from node 2 towards node 3. We obtain the final graph $G''$ as follows:
\begin{itemize}
\item For every node $v \in V_1$, we replace the directed edge $(1,v)$ with a path $1 \to a_{v,1} \to...\to a_{v,p} \to v$ by adding $p$ intermediate nodes. All edges on the path have weight $0$.
\item For every node $v \in V_2$, we replace the directed edge $(v,2)$ with a path $v \to b_{v,1} \to...\to b_{v,p} \to 2$ by adding $p$ intermediate nodes. All edges on the path have weight $0$.
\item Finally, we replace the directed edge $(2,3)$ with a path $2 \to c_{1} \to...\to c_{Z} \to 3$ by adding $Z$ intermediate nodes, where $2^Z$ is a strict upper bound on the weight of any $k$-clique\footnote{A trivial such upper bound is $k^2$ times the weight of the maximum edge.}. All edges on the path have weight $0$.
\end{itemize}

We create an instance of the $\Viterbi$ problem $(A,B,S)$ as follows:
\begin{itemize}
\item Matrix $A$ is the weighted adjacency matrix of $G''$ that takes value $+\infty$ (or a sufficiently large integer) for non-existent edges and non-existent self-loops. 
\item The alphabet of the HMM is $U_1 \cup ... \cup U_p \cup \{ \bot, \bot_0, \bot_1, \bot_F \}$ and thus matrix $B$ has $O(n)$ rows and $\sigma = p \cdot m + 4 = O(n^\alpha)$ columns. 

For all $v \in V_1$, every $i \in [p]$ and every $u \in U_i$, $B(a_{v,i},u)$ is equal to the weight of the edge $(v,u)$ in graph $G$. Similarly, for all $v \in V_2$, every $j \in [p]$ and every $u \in U_j$, $B(b_{v,j},u)$ is equal to the weight of the edge $(v,u)$ in graph $G$. 

Moreover, for all $i \in \{1,...,Z\}$, $B(c_{i},\bot_1) = 2^{i-1}$ and $B(c_{i},\bot_0) = 0$.
Finally, $B(v,\bot) = +\infty$ for all nodes $v \not \in \{1,3\}$ while $B(v,\bot) = +\infty$ for all nodes $v \neq 3$. 
All remaining entries of matrix $B$ are $0$.

\item Sequence $S$ is generated by appending for every tuple $(u_1,...,u_p) \in U_1 \times ... \times U_p$ the following observations in this order: Initially we add the obervations ($u_1$, ..., $u_p$, $\bot_0$, $\bot_0$, $u_1$, ..., $u_p$, $\bot_0$). Moreover, let $W$ be the total weight of the clique ($u_1$, ..., $u_p$) in the graph $G$. We add $Z$ observations encoding $W$ in binary\footnote{Since $2^Z$ is a strict upper-bound on the clique size at most $Z$ digits are required.} starting with the least significant bit. For example, if $W=11$ and $Z=5$, the binary representation is $01011_2$ and the observations we add are $\bot_1, \bot_1, \bot_0, \bot_1, \bot_0$ in that order. Finally, we append a $\bot$ observation at the end.

Notice, that for each tuple, we append exactly $Z+2p+4 = Z+2k$ observations. Thus, the total number of observations is $m^p (Z+2k)$. We add a final $\bot_F$ observation at the end and set $T= m^p (Z+2k)+1$.
\end{itemize}

\paragraph{Correctness of the reduction}
Since the $\kClique$ instance requires \\ $\Omega\left( |V_1| \cdot |V_2| \cdot \prod_{i=1}^p |U_i| \right)^{1-o(1)} = \Omega(Tn^2)^{1-o(1)}$ time, the following claim implies that the above $\Viterbi$ instances require $\Omega(Tn^2)^{1-o(1)}$ time. The alphabet size used is at most $O(n^\alpha)$ and $\alpha \le \eps$ and the theorem follows.

\begin{claim}
The weight of the solution to the $\Viterbi$ instance is equal to the minimum weight of a $k$-clique in the graph $G$.
\end{claim}

\begin{proof}
The optimal path for the $\Viterbi$ instance begins at node $1$. It must end in node $3$ since otherwise when observation $\bot_F$ arrives we collect cost $+\infty$. Similarly, whenever an observation $\bot$ arrives the path must be either on node $1$ or $3$. Thus, the path first loops in node $1$ and then goes from node $1$ to node $3$ during the sequence of $Z+2k$ consecutive observations corresponding to some tuple $(u_1,...,u_p) \in U_1 \times ... \times U_p$ and stays in node $3$ until the end. Let $v_1$ and $v_2$ be the nodes in $V_1$ and $V_2$, respectively, that are visited when moving from node $1$ to node $3$. The only steps of non-zero cost happen during the subsequence of observations corresponding to the tuple $(u_1,...,u_p)$:
\begin{enumerate}
\item When the subsequence begins with $u_1$, the path jumps to node $a_{v_1,1}$ which has a cost $B(a_{v_1,1},u_1)$ equal to the edge-weight $(v_1,u_1)$ in graph $G$. It then continues on to nodes $a_{v_1,2},...,a_{v_1,p}$ when seeing observations $u_2,...,u_p$. The total cost of these steps is $\sum_{i=1}^p B(a_{v_1,i},u_i)$ which is the total weight of edges $(v_1,u_1),...,(v_1,u_p)$ in graph $G$.
\item For the next two observations $\bot_0, \bot_0$, the path jumps to nodes $v_1$ and $v_2$. The first jump has no cost while the latter has cost $A(v_1,v_2)$ equal to the weight of the edge $(v_1,v_2)$ in $G$.
\item The subsequence continues with observations $u_1,...,u_p$ and the path jumps to nodes $b_{v_2,1},...,b_{v_2,p}$ which has a total cost $\sum_{i=1}^p B(b_{v_2,i},u_i)$ which is equal to the total weight of edges $(v_2,u_1),...,(v_2,u_p)$ in graph $G$.
\item The path then jumps to node $2$ at no cost at observation $\bot_0$.
\item The path then moves on to the nodes $c_1,...,c_Z$. The total cost of those moves is equal to the total weight of the clique ($u_1$, ..., $u_p$) since the obervations $\bot_0$ and $\bot_1$ that follow encode that weight in binary.

\end{enumerate}
The overall cost of the path is exactly equal to the weight of the $k$-clique $(v_1,v_2,u_1,...,u_p)$ in $G$. Minimizing the cost of the path in this instance is therefore the same as finding the minimum weight $k$-clique in $G$.
\end{proof}

 	\section{A faster $\Viterbi$ algorithm}
In this section, we present a faster algorithm for the $\Viterbi$ problem, when there are only few distinct transition probabilities in the underlying HMM.

\algorithm*

The number of distinct transition probabilities is equal to the number of distinct entries in matrix $\tilde A$ in Definition~\ref{def:HMM}. The same is true for matrix $A$ in the additive version of $\Viterbi$, in Definition~\ref{def:viterbipath}. So, from the theorem statement we can assume that matrix $A$ has at most $2^{\eps \sqrt {\log n} }$ different entries for some constant $\eps > 0$.

To present our algorithm, we revisit the definition of $\Viterbi$. We want to compute a path $u_0=1,u_1,\dots,u_T$ that minimizes the quantity:
\begin{equation}\label{eq:problem}
\min_{u_0=1,u_1,\dots,u_T} \sum_{t=1}^T \left[  A(u_{t-1},u_{t}) +  B(u_{t},s_t) \right].
\end{equation}

Defining the vectors $b_t = B(\cdot,s_t)$, we note that \eqref{eq:problem} is equal to the minimum entry in the vector obtained by a sequence of $T$ $(\min,+)$ matrix-vector products\footnote{A $(\min,+)$ product between a matrix $M$ and a vector $v$ is denoted by $M \oplus v$ and is equal to a vector $u$ where $u_i = \min_j (M_{i,j} + v_j)$.} as follows:
\begin{equation}\label{eq:minplus}
A \oplus (\ldots(A \oplus ( A \oplus (A \oplus z + b_1) + b_2 )+b_3) \ldots ) + b_T
\end{equation}
where $z$ is a vector with entries $z_1 = 0$ and $z_i = \infty$ for all $i\neq 1$. Vector $z$ represents the cost of being at node $i$ at time $0$. Vector $(A \oplus z + b_1)$ represents the minimum cost of reaching each node at time $1$ after seeing observation $s_1$. After $T$ steps, every entry $i$ of vector~\eqref{eq:minplus} represents the minimum minimum cost of a path that starts at $u_0=1$ and ends at $u_T=i$ after $T$ observations. Taking the minimum of all entries gives the cost of the solution to the $\Viterbi$ instance.

To evaluate~\eqref{eq:minplus}, we design an online $(\min,+)$ matrix-vector multiplication algorithm. In the online matrix-vector multiplication problem, we are given a matrix and a sequence of vectors in online fashion. We are required to output the result of every matrix-vector product before receiving the next vector. Our algorithm for online $(\min,+)$ matrix-vector multiplication is based on a recent algorithm for online Boolean matrix-vector multiplication by Green Larsen and Williams~\cite{larsen2016faster}:

\begin{theorem}[Green Larsen and Williams~\cite{larsen2016faster}] \label{thm:OBMV} For any matrix $M \in \{0,1\}^{n\times n}$ and any sequence of $T=2^{\omega(\sqrt{\log n})}$ vectors $v_1,\ldots,v_T \in \{0,1\}^n$, online Boolean matrix-vector multiplication of $M$ and $v_i$ can be performed in $n^2/2^{\Omega(\sqrt{\log n})}$ amortized time whp. No preprocessing is required.\end{theorem}

We show the following theorem for online $(\min,+)$ matrix-vector multiplication, which gives the promised runtime for the $\Viterbi$ problem\footnote{Even though computing all $(\min,+)$ products does not directly give a path for the $\Viterbi$ problem, we can obtain one at no additional cost by storing back pointers. This is standard and we omit the details.} since we are interested in the case where $T$ and $n$ are polynomially related, i.e. $T = n^{\Theta(1)}$.

\begin{theorem} \label{thm:minplus} Let $A \in \R^{n\times n}$ be a matrix with at most $2^{\eps \sqrt {\log n} }$ distinct entries for a constant $\eps > 0$. For any sequence of $T=2^{\omega(\sqrt{\log n})}$ vectors $v_1,\ldots,v_T \in \R^n$, online $(\min,+)$ matrix-vector multiplication of $A$ and $v_i$ can be performed in $n^2/2^{\Omega(\sqrt{\log n})}$ amortized time whp. No preprocessing is required.\end{theorem}

\begin{proof}
We will show the theorem for the case where $A \in \{0,+\infty\}^{n\times n}$. The general case where matrix $A$ has $d \le 2^{\eps \sqrt {\log n} }$ distinct values $a^1,...,a^d$  can be handled by creating $d$ matrices $A^{1},...,A^{d}$, where each matrix $A^{k}$ has entries $A^{k}_{ij} = 0$ if $A_{ij} = a^k$ and $+\infty$ otherwise. Then, vector $r = A \oplus v$ can be computed by computing $r^k = A^{k} \oplus v$ for every $k$ and setting $r_i = \min_k (r^k_i + a^k)$. This introduces a factor of $2^{\eps \sqrt {\log n} }$ in amortized runtime but the final amortized runtime remains $n^2/2^{\Omega(\sqrt{\log n})}$ if $\eps>0$ is sufficiently small. From now on we assume that $A \in \{0,+\infty\}^{n\times n}$ and define the matrix $\bar A \in \{0,1\}^{n \times n}$ whose every entry is $1$ if the corresponding entry at matrix $A$ is $0$ and $0$ otherwise.

For every query vector $v$, we perform the following:

\begin{itemize}
\item[--] Sort indices $i_1, ...,i_n$ such that $v_{i_1} \le ... \le v_{i_n}$ in $O(n \log n)$ time.
\item[--] Partition the indices into $p = 2^{\alpha \sqrt{\log n}}$ sets, where set $S_k$ contains indices $i_{(k-1) \ceil {\frac n p} + 1},...,i_{k \ceil {\frac n p}}$.
\item[--] Set $r = (\bot,...,\bot)^T$, where $\bot$ indicates an undefined value.
\item[--] For $k = 1 ... p$ fill the entries of $r$ as follows:
\begin{itemize}
	\item[-] Let $\mathbb{I}_{S_k}$ be the indicator vector of $S_k$ that takes value $1$ at index $i$ if $i \in S_k$ and $0$ otherwise.
	\item[-] Compute the Boolean matrix-vector product $\pi^k = \bar A \odot \mathbb{I}_{S_k}$ using the algorithm from Theorem~\ref{thm:OBMV}.
	\item[-] Set $r_j = \min_{i \in S_k} (A_{j,i} + v_i)$ for all $j \in [n]$ such that $r_j = \bot$ and $\pi^k_j = 1$. 
\end{itemize}
\item[--] Return vector $r$.
\end{itemize}

\paragraph{Runtime of the algorithm per query} The algorithm performs $p = 2^{\alpha \sqrt{\log n}}$ Boolean matrix-vector multiplications, for a total amortized cost of $p \cdot n^2 / {2^{\Omega(\sqrt{\log n})}} = {n^2}/ {2^{\Omega(\sqrt{\log n})}}$ for a small enough constant $\alpha > 0$. Moreover, to fill an entry $r_j$ the algorithm requires going through all elements in some set $S_k$ for a total runtime of $O(|S_k|) = {n}/ {2^{\Omega(\sqrt{\log n})}}$. Thus, for all entries $p_j$ the total time required is ${n^2}/ {2^{\Omega(\sqrt{\log n})}}$. The runtime of the other steps is dominated by these two operations so the algorithm takes ${n^2}/ {2^{\Omega(\sqrt{\log n})}}$ amortized time per query.

\paragraph{Correctness of the algorithm} To see that the algorithm correctly computes the $(\min,+)$ product $A \oplus v$, observe that the algorithm fills in the entries of vector $r$ from smallest to largest. Thus, when we set a value to entry $r_j$ we never have to change it again. Moreover, if the value $r_j$ gets filled at step $k$, it must be the case that $\pi^{k'}_j = 0$ for all $k' < k$. This means that for all indices $i \in S_1 \cup ... \cup S_{k-1}$ the corresponding entry $A_{j,i}$ was always $+\infty$.
\end{proof}

\section{Nondeterministic algorithms for the $\Viterbi$ problem}\label{sec:nondeterministic}
It is natural question to ask if one could give a conditional lower bound for the $\Viterbi$ problem by making a different hardness assumption such as the Strong Exponential Time Hypothesis (SETH)\footnote{A direct implication of SETH is that the satisfiability problem of CNFs can't be decided in time $O(2^{n(1-\eps)})$ for any constant $\eps>0$.} or the $3$-Sum conjecture\footnote{The $3$-Sum conjecture states that given a set of $n$ integers, deciding if it contains three integers that sum up to $0$ requires $\Omega(n^{2-o(1)})$ time.}. 
In this section we argue that there is no conditional lower bound for the $\Viterbi$ problem that is based on SETH if the Nondeterministic Strong Exponential Time Hypothesis (NSETH)\footnote{The NSETH is an extension of SETH and it states that deciding the language of unsatisfiable CNFs requires $\Omega(2^{n(1-\eps)})$ \emph{nondeterministic} time.} is true. NSETH was introduced in \cite{carmosino2016nondeterministic} and the authors showed the following statement. If a certain computational problem is in ${\sf NTIME}(n^{c-\alpha}) \cap {\sf coNTIME}(n^{c-\alpha})$ for a constant $\alpha>0$, then, if NSETH holds, there is no $\Omega(n^{c-o(1)})$ conditional lower bound for this problem based on SETH. Thus, to rule out $\Omega(Tn^2)^{1-o(1)}$ SETH-based lower bound for $\Viterbi$, it suffices to show that $\Viterbi \in {\sf NTIME}\left((Tn^2)^{1-\alpha}\right) \cap {\sf coNTIME}\left((Tn^2)^{1-\alpha}\right)$ for a constant $\alpha>0$. This is what we do in the rest of the section.

To show that $\Viterbi \in {\sf NTIME}\left((Tn^2)^{1-\alpha}\right) \cap {\sf coNTIME}\left((Tn^2)^{1-\alpha}\right)$, we have to define a decision version of the $\Viterbi$ problem. A natural candidate is as follows: given the instance and a threshold, decide if the cost of the solution is at most the threshold. Clearly, $\Viterbi \in {\sf NTIME}\left(T\right) \subseteq {\sf NTIME}\left((Tn^2)^{1-\alpha}\right)$ since the nondeterministic algorithm can guess the optimal path and check that the solution is at most the threshold.

We now present an algorithm which shows that $\Viterbi \in {\sf coNTIME}\left((Tn^2)^{1-\alpha}\right)$. Consider the formula~\eqref{eq:minplus}. We can rewrite it by defining a sequence of vectors $v_0,v_1,...,v_T$ such that:
$$v_0 = z \text{ and } v_t = A \oplus v_{t-1} + b_t$$
Thus, we can solve the $\Viterbi$ instance by computing $v_T$ and calculating its minimum entry. Notice that the above can be rewritten in a matrix form as:
\begin{equation}\label{eq:minplusmatrix}
  [v_1,...,v_T] = A \oplus [v_0,...,v_{T-1}] + [b_1, ..., b_T]
\end{equation}
Using this equation, the nondeterministic algorithm can check that the solution to $\Viterbi$ is greater than the threshold by directly guessing all vectors $v_1,...,v_T$ and computing the minimum entry of $v_T$. The algorithm can verify that these guesses are correct by checking that equation~\eqref{eq:minplusmatrix} holds. The most computationally demanding step is verifying that the $(\min,+)$-product of matrices $A$ and $[v_0,...,v_{T-1}]$ is accurate but it is known that $\textsc{MinPlusProduct} \in {\sf NTIME}\left((Tn^2)^{1-\alpha}\right) \cap {\sf coNTIME}\left((Tn^2)^{1-\alpha}\right)$ for some constant $\alpha > 0$ \cite{carmosino2016nondeterministic}. Therefore, it follows that also $\Viterbi \in {\sf NTIME}\left((Tn^2)^{1-\alpha}\right) \cap {\sf coNTIME}\left((Tn^2)^{1-\alpha}\right)$.
 	\section{Complexity of $\Viterbi$ for unary alphabet}\label{sec:noobs}

In this section, we focus on the extreme case of $\Viterbi$ with unary alphabet.

\begin{theorem}
	The $\Viterbi$ problem requires $\Omega(Tn^2)^{1-o(1)}$ time when $T\leq n$ even if the size of the alphabet is $\sigma=1$, assuming the APSP Conjecture.
\end{theorem}

The above theorem follows directly from APSP-hardness of the $\Walk$ problem that we present next.

\begin{theorem}
	The $\Walk$ problem requires $\Omega(Tn^2)^{1-o(1)}$ time when $T\leq n$, assuming the APSP Conjecture.
\end{theorem}

\begin{proof}

We will perform a reduction from the $\MinTriangle$ problem to the $\Viterbi$ problem. In the instance of the $\MinTriangle$ problem, we are given a $3$-partite undirected graph $G=(V_1 \cup V_2 \cup U, \ E)$ with positive edge weights such that $|V_1| = |V_2| = n$, $|U| = m$. We want to find a triangle of minimum weight in the graph $G$. 
To perform the reduction, we define a weighted directed and acyclic graph $G'=(\{1,2\} \cup V_1 \cup V_2 \cup U \cup U', \ E')$. Nodes in $U'$ are in one-to-one correspondence with nodes in $U$ and $|U'|=m$. $E'$ is defined as follows. We add all edges of $G$ between nodes in $U$ and $V_1$ directed from $U$ towards $V_1$ and similarly, we add all edges of $G$ between nodes in $V_1$ and $V_2$ directed from $V_1$ towards $V_2$. Instead of having edges between nodes in $V_2$ and $U$, we add the corresponding edges of $G$ between nodes in $V_2$ and $U'$ directed from $V_2$ towards $U'$.
Moreover, we add additional edges of weight $0$ to create a path $P$ of $m+1$ nodes, starting from node $1$ and going through all nodes in $U$ in some order. Finally, we create another path $P'$ of $m+1$ nodes going through all nodes in $U'$ in the same order as their counterparts on path $P$ and ending at node $2$. These edges have weight $0$ apart from the last one, entering node $2$, which has weight $-C$ (a sufficiently large negative constant)\footnote{Since the definition of $\Walk$ doesn't allow negative weights, we can equivalently set its weight to be 0 and add $C$ to all the other edge weights.}.

We create an instance of the $\Walk$ problem by setting $T=m+4$ and $A$ to be the weighted adjacency matrix of $G'$ that takes value $+\infty$ (or a sufficiently large integer) for non-existent edges and self-loops.

The optimal walk of the $\Walk$ instance must include the edge of weight $-C$ entering node $2$ since otherwise the cost will be non-negative. Moreover, the walk must reach node $2$ exactly at the last step since otherwise the cost will be $+\infty$ as there are no outgoing edges from node $2$. By the choice of $T$, the walk leaves path $P$ at some node $u \in U$, then visits nodes $v_1$ and $v_2$ in $V_1$ and $V_2$, respectively, and subsequently moves to node $u' \in U'$ where $u'$ is the counterpart of $u$ on path $P'$. The total cost of the walk is thus the weight of the triangle $(u,v_1,v_2)$ in $G$, minus $C$. Therefore, the optimal walk has cost equal to the weight of the minimum triangle up to the additive constant $C$.
\end{proof}

Notice that when $T > n$, the runtime of the Viterbi algorithm is no longer optimal. Equation~\ref{eq:minplus} for the general $\Viterbi$ problem reduces, in the case of unary alphabet, to computing $(\min,+)$ matrix-vector product $T$ times: $A \oplus A \oplus ... \oplus A \oplus z$. However, this can equivalently be performed by computing $A^{\oplus T}$ using exponentiation with repeated squaring. This requires only $O(\log T)$ matrix $(\min,+)$-multiplications. Using the currently best algorithm for $(\min,+)$ matrix product~\cite{williams2014faster}, we get an algorithm with a total running time $\log T \cdot n^3 / 2^{\Omega(\sqrt{\log n})}$. 	
	\section{Acknowledgments}
We thank Piotr Indyk for many helpful discussions. Arturs Backurs was supported by NSF and Simons Foundation. 
	\bibliographystyle{alpha}
	\bibliography{ref}
	
	\appendix
	
	\section{Sum of Probabilities..........} \label{app:sumprob} \todo{finish this} 
\end{document}